\documentclass[letterpaper, 10 pt, conference]{ieeeconf}  

\IEEEoverridecommandlockouts                              

\usepackage{amsmath}
\usepackage{amssymb}  
\usepackage{amsfonts}
\usepackage{mathtools}
\usepackage{bm}
\usepackage{color}
\usepackage{graphics} 
\usepackage{graphicx}
\usepackage{epsfig} 
\usepackage{epstopdf}
\usepackage[noadjust]{cite}
\usepackage{tikz}
\usetikzlibrary{calc,positioning,shapes,shadows,arrows,fit}
\usepackage{fancyhdr}
\usepackage{fancyref}
\usepackage{hyperref}
\usepackage{float}
\usepackage[font = footnotesize]{caption}
\usepackage{algorithm}
\usepackage{algorithmic}

\usepackage{amsthm}

\newtheorem{theorem}{Theorem}
\newtheorem{lemma}{Lemma}

\theoremstyle{definition}
\newtheorem{definition}{Definition}
\newtheorem{remark}{Remark}
\newtheorem{assumption}{Assumption}

\newtheorem{problem}{Problem}

\title{\LARGE \bf
Design and Experimental Validation of Tube-based MPC for Timed-constrained Robot Planning}

\author{Alexandros Nikou, Shahab Heshmati-alamdari and Dimos V. Dimarogonas
\thanks{The authors are with the School of Electrical Engineering and Computer Science, KTH Royal Institute of Technology, Stockholm, Sweden. E-mail: {\tt\small \{anikou, shaha, dimos\}@kth.se}. This work was supported by the H2020 ERC Grant BUCOPHSYS, the EU H2020 Co4Robots project, the Swedish Foundation for Strategic Research (SSF), the Swedish Research Council (VR) and the Knut och Alice Wallenberg Foundation (KAW).}
}

\begin{document}
\maketitle
\thispagestyle{empty}
\pagestyle{empty}

\begin{abstract}
This paper deals with the design and experimental validation of a state-of-the art tube-based Model Predictive Control (MPC) for achieving time-constrained tasks. Given the uncertain nonlinear dynamics of the robot as well as a high-level task written in Metric Interval Temporal Logic (MITL), the goal is to design a feedback control law that guarantees the satisfaction of the task. The workspace is divided into Regions of Interest (RoI) and contains also unsafe regions (obstacles) that the robot should not visit. The feedback control law consists of two terms: a control input which is the outcome of a Finite Horizon Optimal Control (FHOCP); and a state feedback law that guarantees that the nominal trajectories are bounded within a tube centered along the nominal trajectories. The aforementioned control law guarantees that the robot is safely navigated through the RoI within certain time bounds. The proposed framework can handle the rich expressiveness of MITL and is experimentally tested with a Nexus mobile robot in our lab facilities. The experimental results show that the proposed framework is promising for solving real-life robotic as well as industrial problems.
\end{abstract}

\section{Introduction}

Over the last years the field of controlling systems under formal verification constraints has been gaining significant research attention due to important applications in robotics, autonomous driving and industrial automation.  The timed logic that has primarily been used is Metric Interval Temporal Logic \cite{alur_1994, alex_2016_acc, alex_licenciate,  alex_automatica_2018}. The control synthesis under Metric Interval Temporal Logic (MITL) consists of three parts:
\begin{enumerate}
\item first, the dynamics of the robot are abstracted into a Weighted Transition System (WTS) by providing feedback control laws that can drive the robot between states; the time duration that the robot needs to navigate between the states is modeled as a weight to the transition system.
\item a product between the WTS and an automaton which accepts all words that satisfy the given formula is computed.
\item once a run is found in the product automaton, it maps back into a sequence of feedback laws that satisfy the given formula.
\end{enumerate}

In practical applications, regarding the first part, the feedback control laws need to be appropriately designed such that the following specifications are taken into consideration: (a) state and input constraints; (b) obstacle avoidance; and (c) robustness against potential external disturbances and uncertainties. 

A control technique that has been used for navigation of robotic agents with guaranteeing obstacle avoidance is the potential fields approach \cite{KODITSCHEK1990412}. However, input constraints and potential external disturbances cannot be incorporated in the control design in a straightforward manner. In the same context, Prescribed Performance Controllers (PPC) that have been used for robotic navigation (see \cite{babis_automatica}) cannot handle input constraints as well as obstacle avoidance guarantees. Regarding real-time experiments, which is the main focus of the manuscript, both of the aforementioned methodologies require a significant amount of efforts in tuning the control gains. In particular, the potential fields approach requires computation of complicated formulas that consist of derivatives which might lead to numerical instabilities when applied to real platforms.

\begin{figure}[t!]
\centering
\includegraphics[scale = 0.47]{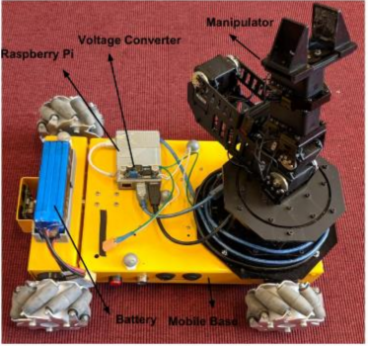}
\caption{A Nexus $10011$ mobile robot with an attached 3 DoF manipulator.}\label{fig:nexus}
\end{figure}

A feedback control law that has been recently proven to be efficient in incorporating the aforementioned specifications is the so-called tube-based Model Predictive Control (MPC) (see \cite{rakovic_2004_tubes_1, yu_2013_tube, alex_ACC_2018, alex_IJRNC_2018}). The idea of the tube-based MPC is based on the fact that the control law has two parts: a control law that is computed on-line and is the outcome of a FHOCP; and a state-feedback law designed offline and guarantees that the real system trajectories are always bounded in a tube whose volume depends on the bound of the disturbances as well as bounds of the derivatives of the dynamics. Experiments with MPC have been provided in \cite{shaha1, shaha2, shaha3, alex_shahab_ifac} without incorporating any tube guarantees.

The contribution of this paper is to experimentally validate our recent theoretical results \cite{alex_2016_acc, alex_ACC_2018, alex_IJRNC_2018} in order to solve a time-constrained planning problem for  a Nexus $10011$ mobile robot (see Fig. \ref{fig:nexus}). In particular, inside of the workspace there exist a set of Regions of Interest (RoI) some of which may be unsafe regions (obstacles). Given an MITL formula that the Nexus is desired to satisfy, the three steps control synthesis MITL framework is demonstrated in practice by application of a sequence of tube-based MPC laws that guarantees the satisfaction of the formula. The experimental results of the paper in hand verify the efficiency of the proposed framework which solves a general category of time-constrained robot navigation problems under state/input constraints, obstacle avoidance as well as uncertainties/disturbances. To the best of the authors' knowledge, this is the first time that such a general control synthesis framework in terms of dynamics, language expressiveness and state/input constraints is demonstrated. The experimental results show that the proposed framework is promising for solving real-life robotic as well as industrial problems.

\section{Notation and Preliminaries} \label{sec:notation_preliminaries}

The $n$-fold Cartesian product of a set $\mathcal{S}$ and its cardinality are defined by $\mathcal{S}^{n}$ and $|\mathcal{S}|$, respectively; $\|y\|_{\scriptscriptstyle 2} \coloneqq \sqrt{y^\top y}$ and $\|y\|_{\scriptscriptstyle M}$ $\coloneqq \sqrt{y^\top M y}$, $M $ $\ge 0$ stand for the Euclidean and the weighted norm of a vector $y \in \mathbb{R}^n$, respectively; $\lambda_{\scriptscriptstyle \min}(M)$ stands for the minimum absolute value of the real part of the eigenvalues of $M \in \mathbb{R}^{n \times n}$; $0_{m \times n} \in \mathbb{R}^{m \times n}$ and $I_n \in \mathbb{R}^{n \times n}$ stand for the $m \times n$ matrix with all entries zeros and the identity matrix, respectively; The set $\mathcal{M}(\chi, r) \coloneqq \{z \in \mathbb{R}^n: \|z-\chi\|_{2} \leq r\}$ models a ball with center and radius $\chi \in \mathbb{R}^{n}$, $r >0$, respectively. Given~the~sets~$\mathcal{S}_1$, $\mathcal{S}_2$~$\subseteq \mathbb{R}^n$~and~the~matrix $M \in \mathbb{R}^{n \times m}$,~ the \emph{matrix-set multiplication}, the \emph{Minkowski addition} and the~\emph{Pontryagin~difference} are respectively defined by: $M \circ \mathcal{S} \coloneqq \{m: \exists s \in \mathcal{S}, m = Ms\}$, $\mathcal{S}_1 \oplus \mathcal{S}_2 \coloneqq \{s_1 + s_2 : s_1 \in \mathcal{S}_1, s_2 \in \mathcal{S}_2\}$,~and~$\mathcal{S}_1 \ominus \mathcal{S}_2 \coloneqq \{s_1 : s_1+s_2 \in \mathcal{S}_1, \forall s_2 \in \mathcal{S}_2\}$.

\begin{definition} \label{def:RPI_set} \cite{yu_2013_tube}
Consider a dynamical system $\dot{x} = f(x)+g(x)u +\delta$ where: $x \in \mathcal{X}$, $u \in \mathcal{U}$ and $\delta \in \Delta$. Consider a set $\mathcal{Q} \subseteq \mathcal{X}$. If there exists a feedback control law $u \coloneqq \kappa(x) \in \mathcal{U}$, such that for all $x(0) \in \mathcal{Q}$ and for all $\delta(t) \in \Delta$ it holds that $x(t) \in \mathcal{Q}$ for all $t \ge 0$, along every solution $x(t)$, then $\mathcal{Q}$ is called a \emph{Robust Control Invariant (RCI) set} for the system.
\end{definition}

\begin{definition} \label{def: WTS} \cite{alex_2016_acc}
A \emph{Weighted Transition System} (WTS) is a tuple $(S, S_0, {\rm Act}, \longrightarrow, \mathfrak{t}, \Gamma, L)$ where $S$ is a set of states; $S_0 \subseteq S$ is a set of initial states; ${\rm Act}$ is a set of actions; $\longrightarrow \subseteq S \times {\rm Act} \times S$ is a transition relation; $\mathfrak{t}: \longrightarrow \rightarrow \mathbb{Q}_{+}$ is a function that maps a weight to each transition; $\Gamma$ is a set of atomic propositions; and $L: S \rightarrow 2^{\Gamma}$ stands for the labeling function.
\end{definition}

\begin{definition}\label{run_of_WTS} \cite{alex_2016_acc}
A \textit{timed run} of a WTS is an infinite sequence $r^t = (r(0), \tau(0))(r(1), \tau(1)) \ldots$, such that $r(0) \in S_0$, and for all $l \geq 0$, it holds that $r(l) \in S$ and $(r(l), u(l), r(l+1)) \in \longrightarrow$ for a sequence of actions $u(0) u(1) u(2) \ldots$ with $u(l) \in \rm Act$, $\forall l \geq 0$. The \textit{time stamps} $\tau(l)$, $l \geq 0$ are inductively defined as: $1)$ $\tau(0) = 0$; $2)$ $\displaystyle \tau(l+1) \coloneqq  \tau(l) + \mathfrak{t}(r(l), r(l+1))$, $\forall l \geq 0$.
\end{definition}

The Metric Interval Temporal Logic (MITL) \cite{alur_1994} over a set of atomic propositions $\Gamma$ is defined by the grammar:
\begin{equation*}
\varphi := \gamma \ | \ \neg \varphi \ | \ \varphi_1 \wedge \varphi_2 \ | \ \bigcirc_I \varphi  \ | \ \Diamond_I \varphi \mid \square_I \varphi \mid  \varphi_1 \ \mathcal{U}_I \ \varphi_2,
\end{equation*}
where $\gamma \in \Gamma$, and $\bigcirc$, $\Diamond$, $\square$ and $\mathcal{U}$ are the next, eventually, always and until temporal operator, respectively; $\neg$, $\wedge$ are the negation and conjunction operators, respectively; $I$ stands for a non-empty timed interval. For the semantics of MITL see \cite{alur_1994}. Technical details regarding timed verification can be found in \cite{bouyer2009qualitative}.

\section{Problem Formulation} \label{sec:problem_formulation}
\subsection{System Model}

Consider a robot operating in a bounded workspace $\mathcal{W} \subseteq \mathbb{R}^{n}$ with \emph{uncertain dynamics:}
\begin{equation} \label{eq:dynamics}
\dot{x} = f(x) + g(x) u + \delta,
\end{equation}
where $f: \mathbb{R}^{n} \to \mathbb{R}^{n}$, $g: \mathbb{R}^{n} \to \mathbb{R}^{n\times n}$ are known and continuously differentiable functions; $u \in \mathbb{R}^{n}$ stands for the control input; and $\delta \in \mathbb{R}^{n}$ models the external disturbances and uncertainties. Consider also input constraints as well as bounded disturbances:
\begin{align*}
u \in \mathcal{U} & \coloneqq \{u \in \mathbb{R}^{n} : \|u\|_{2} \le \widetilde{u}\}, \\
\delta \in \Delta & \coloneqq \{\delta \in \mathbb{R}^{n} : \|\delta\|_{2} \le \widetilde{\delta}\},
\end{align*}
where $\widetilde{u}$, $\widetilde{\delta} > 0$ are a priori known. Define the corresponding \emph{nominal dynamics} for \eqref{eq:dynamics} by:
\begin{equation*}
\dot{\hat{x}} = f(\hat{x}) + g(\hat{x}) \hat{u},
\end{equation*}
which are the dynamics for the case of $\delta = 0$.

\begin{assumption} \label{ass1}
Let the dynamics:
$$\dot{\hat{x}} = A \hat{x} + B \hat{u}.$$
be the Jacobian linearization of the nominal dynamics \eqref{eq:dynamics} around the equilibrium state $x = 0$. Then, we assume that the latter system is stabilizable. 
\end{assumption}

\begin{assumption} \label{ass2}
There exists a strictly positive constant $\underline{g}$ such that the following holds:
\begin{align*}
\hspace{-3mm} \lambda_{\min}\left[\tfrac{g(x) + g(x)^\top}{2}\right] \ge \underline{g} > 0, \ \ \forall x \in \mathcal{W}. \hspace{-3mm}
\end{align*}
Furthermore, it holds that $f(0) = 0$ and $g(0) = 0$.
\end{assumption}

In the given workspace, there exist $z \in \mathbb{N}$ Regions of Interest (RoI) labeled by $\mathcal{Z} \coloneqq \{1, \dots, z\}$. The RoI are modeled by the balls $\mathcal{R}_z  \coloneqq \mathcal{M}(y_z, p_z)$, $z \in \mathcal{Z}$, where $y_z$ and $p_z > 0$ stands for the center and radius of RoI $\mathcal{R}_z$, respectively. At each time $t \ge 0$, the robot is occupying a ball $\mathcal{M}(x(t), \eta)$ that covers its volume, where $x(t)$ and $\eta > 0$ are its center and radius, respectively.

\subsection{Objectives}

The main goal of this paper is to design a feedback control law that steers the robot with dynamics as in \eqref{eq:dynamics} between RoI so that it obeys a high-level task given in MITL. Define the labeling function:
\begin{align} \label{eq:label_function}
L: \bigcup_{z \in \mathcal{Z}} \mathcal{R}_z \to 2^{\Gamma},
\end{align}
which maps each RoI with a subset of atomic propositions.

\begin{figure*}
	\centering
	\begin{tikzpicture}[scale = 0.75] 
	\draw[blue!70, line width=.04cm] (-16.8, 6.0) rectangle +(2.7, 0.9);
	\node at (-15.45, 6.45) {$\text{MITL2TBA}$};
	
	\draw[-latex, draw=black, line width = 1.0] (-15.5,6.0) -- (-15.5,4.6);
	\draw[-latex, draw=black, line width = 1.0] (-15.5,7.6) -- (-15.5,6.9);
	
	\node at (-15.5, 8.00) {$\varphi$};
	\node at (-15.0, 5.35) {$\mathcal{A}$};
	\node at (-15.5, 4.45) {$\otimes$};
	
	
	\draw[-latex, draw=black, line width = 1.0] (-15.20,4.45) -- (-14.0,4.45);
	
	\node at (-13.6, 4.47) {$\widetilde{\mathcal{T}}$};
	
	\draw[-latex, draw=black, line width = 1.0] (-16.60,4.45) -- (-15.7,4.45);
	\node at (-17.2, 4.47) {$\mathcal{T}$};
	
	\draw[-latex, draw=black, line width = 1.0] (-13.20,4.45) -- (-12.5,4.45);
	\draw[red!70, line width=.04cm] (-12.5, 4.15) rectangle +(2.20, 0.70);
	\node at (-11.32, 4.50) {$\text{synthesis}$};
	
	\draw[-latex, draw=black, line width = 1.0] (-10.25,4.45) -- (-9.75,4.45);
	
	\node at (-9.35, 4.47) {$\widetilde{r}$};
	
	\draw[-latex, draw=black, line width = 1.0] (-18.45,4.47) -- (-17.90,4.47);
	\draw[orange!70, line width=.04cm] (-21.0, 4.20) rectangle +(2.5, 0.7);
	\node at (-19.70, 4.60) {$\text{abstraction}$};
	
	\draw[-latex, draw=black, line width = 1.0] (-21.70,4.47) -- (-21.00,4.47);
	
	\node at (-24.0, 4.50) {$\displaystyle \dot{x}= f(x)+g(x)u +\delta$};
	
	\draw [black, line width = 0.030cm] (-9.35, 4.80) -- (-9.35, 8.50);
	\draw [black, line width = 0.030cm] (-9.35, 8.50) -- (-18.90, 8.50);
	\draw [black, line width = 0.030cm] (-20.20, 8.50) -- (-24.00, 8.50);
	\draw[-latex, draw=black, line width = 1.0] (-24.00, 8.50) -- (-24.00, 5.1);
	
	\draw[green!70, line width=.04cm] (-20.2, 7.95) rectangle +(1.3, 1.0);
	\node at (-19.5, 8.40) {$u$};
	\end{tikzpicture}
	\caption{A graphical representation of control design framework.}
	\label{fig:solution_scheme}
\end{figure*}
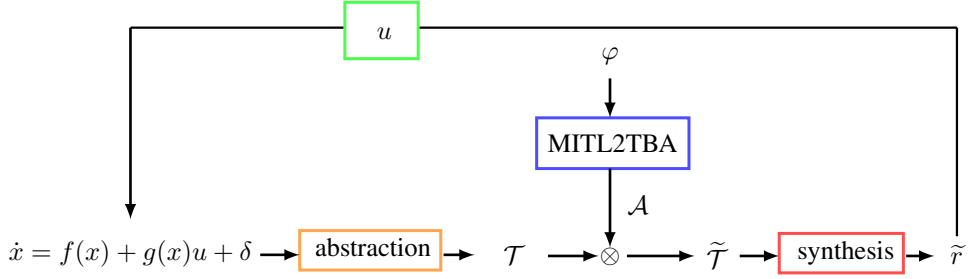

\begin{definition} \label{def:unique_timed_word}
	A trajectory $x(t)$ is \emph{uniquely associated with a timed run} $r^t = (r(0), \tau(0))$ $(r(1), \tau(1))\ldots$, if:	
	\begin{enumerate}
		\item $\tau(0) = 0$, i.e., the robot starts its motion at time $t = 0$;
		\item $r(l) \in \displaystyle \bigcup_{z \in \mathcal{Z}} \mathcal{R}_z$, for every $l \in \mathbb{N}$;
		\item $\mathcal{M}(x(\tau(0)), \eta) \subsetneq r(0)$, i.e., initially, the volume of the robot is entirely within the RoI $r(0)$;
		\item $\mathcal{M}(x(\tau(l)), \eta) \subsetneq r(l)$, $\forall l \in \mathbb{N}$, i.e., the robot changes discrete state only when its entire volume is contained in the corresponding RoI;
		\item $\tau(l+1) \coloneqq \tau(l) + \mathfrak{t}(r(l), r(l+1))$, $\forall l \in \mathbb{N}$, where:
		\begin{equation} \label{eq:desired_times_T}
		\mathfrak{t}: \left( \bigcup_{z \in \mathcal{Z}} \mathcal{R}_z \right) \times \left( \bigcup_{z \in \mathcal{Z}} \mathcal{R}_z \right) \to \mathbb{Q}_{+},
		\end{equation}
		is a function that models the time duration that the robot needs to be steered between regions $r(l)$ and $r(l+1)$.
	\end{enumerate}
\end{definition}  

\begin{definition} \label{def:x_satisfaction}
A trajectory $x(t)$ \emph{satisfies} a formula $\varphi$ written in MITL over a set of atomic propositions $\Gamma$ (written as $x(t) \models \varphi$, $\forall t \ge 0$), if and only if there exists a timed run $r^t$ to which the trajectory is uniquely associated, according to Definition \ref{def:unique_timed_word}, which satisfies $\varphi$.
\end{definition}

\begin{problem} \label{problem}
Consider a robot with dynamics as in \eqref{eq:dynamics}, operating in a workspace $\mathcal{W} \subseteq \mathbb{R}^{n}$. In the workspace, there exist $z \in \mathbb{N}$ RoI modeled by the balls $\mathcal{M}(y_z, p_z)$, $\forall z \in \mathcal{Z}$. Then, given an MITL formula $\varphi$ over a set of atomic propositions $\Gamma$ and a labeling function as in \eqref{eq:label_function}, design a feedback control law $u = \kappa(x) \in \mathcal{U}$ which guarantees that:
\begin{align*}
x(t) \models \varphi, \forall t \ge 0,
\end{align*}
according to Definition \ref{def:x_satisfaction}, while the robot remains in the workspace for all times.
\end{problem}

\section{Problem Solution} \label{sec:problem_solution}

\subsection{Feedback Control Design}

Consider a robot with dynamics \eqref{eq:dynamics} occupying a RoI $\mathcal{R}_i$, $i \in \mathcal{Z}$ at time $\mathfrak{t}_i \ge 0$. Denote by $x_j \in \mathcal{R}_j$, $j \in \mathcal{Z} \backslash \{i\}$ the center of a desired RoI towards which the robot is required to be navigated. Define the error vector $e \coloneqq x - x_i \in \mathbb{R}^n$. Then, the \emph{uncertain error dynamics} are given by:
\begin{align} \label{eq:uncrt_error_dynamics}
\dot{e} = f(e+x_i)+g(e+x_i) u + \delta.
\end{align}
The corresponding \emph{nominal error dynamics} are given by:
\begin{align} \label{eq:nominal_error_dynamics}
\dot{\hat{e}} = f(\hat{e}+x_i)+g(\hat{e}+x_i) \hat{u}
\end{align}
Define the set that captures the state constraints by:
\begin{align*}
\mathcal{X} \coloneqq \{x \in \mathbb{R}^n : & \ \mathcal{M}(x(t), \eta) \cap \mathcal{R}_{j'} = \emptyset, \ \ \forall j' \in \mathcal{Z} \backslash \{i, j\}, \\ 
&\hspace{1mm} \mathcal{M}(x(t), \eta) \subsetneq \mathcal{W} \}.
\end{align*}
The first constraint denotes the fact that the robot should not intersect with any other RoI other than $\mathcal{R}_i$ ,$\mathcal{R}_j$; the second one, denotes the fact that the robot needs to remain in the workspace for all times.

Define by $q = e - \hat{e}$ the deviation between the real state of the system \eqref{eq:uncrt_error_dynamics} and the nominal state of the system \eqref{eq:nominal_error_dynamics} with $q(0) = e(0) -\hat{e}(0) = 0$. The dynamics of the state $q$ are given by:
\begin{align}
\dot{q} & =  f(e+x_i)- f(\hat{e}+x_i) \notag \notag \\
&\hspace{19mm} +g(e+x_i) u-g(\hat{e}+x_i) \hat{u} + \delta \notag \\
& = h(e, \hat{e}, u) + g(\hat{e}+x_i)(u-\hat{u})+\delta. \label{eq:q_dynamics}
\end{align}
where the function $h$ is defined by $h(e, \hat{e}, u) \coloneqq f(e+x_i)- f(\hat{e}+x_i) + g(e+x_i)u-g(\hat{e}+x_i)u$. Note that the following holds:
\begin{align*}
\|h(e, \hat{e}, u)\|_2 \le L \|q\|_2, 
\end{align*}
where $L \coloneqq \max\{L_f, L_g\}$ and $L_f$, $L_g > 0$ are the Lipschitz constant of the functions $f$, $g$, respectively.

\begin{lemma}
The feedback control law:
\begin{align} \label{eq:control_input_u}
u \coloneqq \hat{u}(\hat{e}) - \sigma q,
\end{align}
where $\hat{u}$ is a nominal input to be computed afterwards through an on-line optimal control problem and the control gain is designed such that $\sigma = \frac{L}{\underline{g}} + \underline{\sigma}$, $\underline{\sigma} > 0$, renders the set:
\begin{align*}
\mathcal{Q} \coloneqq \left\{q \in \mathbb{R}^n : \|q\|_2 \le \frac{\widetilde{\delta}}{\underline{\sigma}} \right\},
\end{align*}
an RPI set for the system \eqref{eq:q_dynamics}, according to Definition \ref{def:RPI_set}.
\end{lemma}
\begin{proof}
The proof of this lemma follows similar arguments to \cite{alex_ACC_2018, alex_IJRNC_2018}. The time derivative of the function $\Lambda(q) = \frac{1}{2} \|q\|^2$ along the trajectories of the system \eqref{eq:q_dynamics} is:
\begin{align*}
\dot{\Lambda}(q) & = q^\top \dot{q} \\
& = q^\top h(e, \hat{e}, u) + q^\top g(e+x_i) (u-\hat{u}) + q^\top \delta \\
& \le L \|q\|_2^2 - \sigma q^\top g(e+x_i) q + \widetilde{\delta} \|q\|_2 \\
& \le  \left[-\left(\sigma \underline{g} - L \right) \|q\|_2 +\widetilde{\delta} \right] \|q\|_2 \notag \\
& \le  \left(-\underline{\sigma} \|q\|_2 +\widetilde{\delta} \right) \|q\|_2.
\end{align*}
Thus, $\dot{\Lambda} < 0$ when $\|q\|_2 > \frac{\widetilde{\delta}}{\underline{\sigma}}$. Since the fact that $q(0) = 0$, it holds that $\|q(t)\| \le \frac{\widetilde{\delta}}{\underline{\sigma}}$, for every $t \ge 0$.
\end{proof}

\begin{remark}
The volume of the tube that is centered along the nominal trajectories $\hat{e}(t)$, $t \ge 0$ depends on the upper bound of the disturbance $\widetilde{\delta}$, and the parameters $L$, $\underline{g}$.
\end{remark}

Hereafter we introduce the methodology under which the online control law $\hat{u}(\hat{e})$ is calculated. Denote by $h > 0$ and $N > h$ the sampling step and the finite prediction horizon. Consider a sequence of sampling times ${t_k}$, $k \in \mathbb{N}$. Then, at every sampling time $t_k$, $k \in \mathbb{N}$ the following FHOCP is solved:
\begin{subequations}
\begin{align}
&\hspace{-4mm}\min\limits_{\hat{u}(\cdot)} \left\{  \|\hat{e}(t_k+N)\|^2_{\scriptscriptstyle P} \hspace{-1mm} + \hspace{-2mm}\int_{t_k}^{t_k+N} \hspace{-1mm}\Big[ \|\hat{e}(s)\|^2_{\scriptscriptstyle Q} +\|\hat{u}(s)\|^2_{\scriptscriptstyle R} \Big] ds \right\} \hspace{0mm} \label{eq:mpc_cost_function} \hspace{-7mm}\\
&\hspace{-4mm}\text{subject to:} \notag \\
&\hspace{-3mm} \dot{\hat{e}}(s) = f(\hat{e}(s)+x_i) +g(\hat{e}(s)+x_i)\hat{u}(s), \ \ \hat{e}(t_k) = e(t_k), \label{eq:diff_mpc} \\
&\hspace{-3mm} \hat{e}(s) \in \mathcal{E} , \ \ \hat{u}(s) \in \mathbb{U},  \ \ s \in [t_k,t_k+T], \label{eq:mpc_constrained_set} \\
&\hspace{-3mm} \hat{e}(t_k+N)\in \mathcal{F}, \label{eq:mpc_terminal_set}
\end{align}
\end{subequations}
where $Q$, $P$ and $R$ are positive definite gain matrices to be appropriately tuned. The state and input constraints sets are modified as:
\begin{align*} 
\mathcal{E} \coloneqq \left[\mathcal{X} \oplus (-x_i)\right] \ominus \mathcal{Q},\ \ \mathbb{U} \coloneqq \mathcal{U} \ominus \left[-\sigma \circ \mathcal{Q} \right],
\end{align*}
in order to guarantee that while the FHOCP is solved for the nominal system dynamics \eqref{eq:nominal_error_dynamics}, the real trajectories satisfy the state and input constraints $\mathcal{X} \oplus (-x_i)$, $\mathcal{U}$, respectively. The set $\mathcal{F} \coloneqq \{\hat{e} \in \mathcal{E} : \|\hat{e}\|_{P} \le \varepsilon\}$, $\varepsilon > 0$ is used in order to ensure the nominal stability of the system \cite{frank_1998_quasi_infinite}.

\begin{theorem} (\cite{alex_ACC_2018})
Assume that Assumptions \ref{ass1} and \ref{ass2} hold. Consider a time $\mathfrak{t}_i \ge 0$ that the robot occupies a RoI $\mathcal{R}_i$, $i \in \mathcal{Z}$ with center $x_i$. Let $x_j$, $j \neq i$ be the center of a desired RoI $\mathcal{R}_j$ that the robot is required to be navigated towards. Suppose also that the FHOCP \eqref{eq:mpc_cost_function}-\eqref{eq:mpc_terminal_set} is feasible at time $\mathfrak{t}_i$. Then, there exist a time $\mathfrak{t}_{ij} > \mathfrak{t}_i$ such that:
\begin{align*}
\|x_i(t)-x_j\| \le \frac{\varepsilon}{\sqrt{\lambda_{\max}(P)}}+\frac{\widetilde{\delta}}{\underline{\sigma}}, \ \ \forall t \ge \mathfrak{t}_{ij}.
\end{align*}
\end{theorem}

According to the previous theorem, the robot with dynamics as in \eqref{eq:dynamics}, starting its motion at time $\mathfrak{t}_i$ at RoI $\mathcal{R}_i$, under the control law \eqref{eq:control_input_u}, will have been navigated to RoI $\mathcal{R}_j$ at time $\mathfrak{t}_{ij}$. For the computation of the time $\mathfrak{t}_{ij}$, Algorithm $1$ is used. In particular, the time that the robot is reaching the terminal set is captured by $\mathfrak{t}_{ij}$. This time models the time the the robot will have been navigated from RoI $\mathcal{R}_i$ to $\mathcal{R}_j$, i.e, is consistent with \eqref{eq:desired_times_T}.

\begin{algorithm}[t!]
	\caption{Numerical computation of $\mathfrak{t}_{\scriptscriptstyle \rm ij} \coloneqq \mathfrak{t}(\mathcal{R}_{\scriptscriptstyle \rm i}, \mathcal{R}_{\scriptscriptstyle \rm j})$}
	\begin{algorithmic}[1]
		\STATE \textbf{Input}: $\mathfrak{t}_{\scriptscriptstyle \rm i}$, $\hat{x}(t_k)$, $x_j$, $k = \{0, 1, 2, \dots\}$;
		\STATE \textbf{Output}:  $\mathfrak{t}_{\scriptscriptstyle \rm ij}$;
		\STATE $t_k \leftarrow \mathfrak{t}_{\scriptscriptstyle \rm i}$;
		\STATE $\rm cond \leftarrow \rm True$
		\WHILE {$\rm cond = \rm True$}
		\STATE measure $\hat{x}(t_k)$;
		\IF {$\|\hat{x}(t_k)-x_{j}\|_2 \le \tfrac{\varepsilon}{\sqrt{\lambda_{\scriptscriptstyle \min}(P)}}$} 
		\STATE $\rm cond \leftarrow \rm False$; 
		\STATE $\rm break$;
		\STATE \textbf{Go to} ``line $13$"  
		\ENDIF
		\STATE $t_k \leftarrow t_k + h$; 
		\ENDWHILE		
		\STATE $\mathfrak{t}_{\scriptscriptstyle \rm ij} \leftarrow t_k$;
	\end{algorithmic} 
\end{algorithm}

\subsection{Discrete Abstractions and Control Synthesis}

\begin{definition} \label{def:WTS_abstraction}
The motion of the robot in the workspace is modeled through a WTS $\mathcal{T} = (S, S_0, {\rm Act}, \longrightarrow, \mathfrak{t}, \Gamma, L)$ where:
\begin{itemize}
\item $S = \bigcup_{z \in \mathcal{Z}} \mathcal{R}_z$ is the set of states;
\item $S_0$ is the initial RoI that the robot starts its motion;
\item $\rm Act$ is set of actions containing the union of all feedback control laws of the form \eqref{eq:control_input_u} that are able to drive the robot between RoI;
\item $\longrightarrow \subseteq S \times Act \times S$ is the transition relation. We say that $(\mathcal{R}_i, u', \mathcal{R}_j) \in \longrightarrow$, $i$, $j \in \mathcal{Z}$, $i \neq j$, if there exist a feedback control $u' \in {\rm Act}$ that can drive the robot from RoI $\mathcal{R}_i$ to the RoI $\mathcal{R}_j$;
\item and $L$, $\mathfrak{t}$ as given in \eqref{eq:label_function} and \eqref{eq:desired_times_T}, respectively.
\end{itemize}
\end{definition}

By using the framework depicted in Fig. \ref{fig:solution_scheme}, a sequence of feedback laws as in \eqref{eq:control_input_u} that guarantee the satisfaction of the given MITL formula $\varphi$ is designed. In particular, the dynamics of the robot \eqref{eq:dynamics} are abstracted into a WTS $\mathcal{T}$ according to Definition \ref{def:WTS_abstraction}; A product WTS $\widetilde{T}$ is constructed by computing the product between $\mathcal{T}$ and the Timed B\"uchi automaton $\mathcal{A}$, whose accepting runs are the ones that satisfy the formula $\varphi$. Then, by performing graph search in $\widetilde{T}$ a timed run $\widetilde{r}$ that maps into feedback control laws as in \eqref{eq:control_input_u} can be computed. For the definitions of the product WTS and the Timed B\"uchi Automaton see \cite{alex_2016_acc}.

\section{Experimental Setup and Results} \label{sec:exp_results}

This section demonstrates the efficacy of the proposed framework via a real-time experiment employing a Nexus $10011$ mobile robot (Fig. \ref{fig:nexus}). The experiment was carried out at Smart Mobility Lab (SML) (see Fig. \ref{fig:nexus_lab} and \cite{SML}). The robot dimensions are $400 \times 360 \times 100$ mm and consists of $4$ aluminum mecanuum wheels which provide omni-directional capabilities through the $3$ degrees of freedom (moving forward/backward, left/right and rotation). The rollers have a rotation of $45$ degrees with reference to the the plane of the wheel. Each wheel is connected to  $12 V$ motors which they have optical encoders. The speed can be controlled by local PID controllers on an Arduino $328$ Controller and Arduino IO expansion board.

SML provides a motion capture system (MoCap) with $12$ Qualisys cameras spread across the lab. The MoCap provides the robot state vector, including pose, orientation as well as linear and angular velocities at frequency of $100$Hz. The hardware used in experiments are connected using Gigabit Ethernet connections, USB connections, and wireless 5.8Ghz connections. The Raspberry Pi on the robot is connected with a TP-Link Router via wireless $5.8$ Ghz connections. The host computer as well as the MoCap are connected with the TP-Link Router via Gigabit Ethernet connections. The host computer runs the node of the controller, which receives the measurement from MoCap, and  calculates the control signal. 

The software implementation of the proposed control strategy was conducted in C++ under the Robot Operating System (ROS) \cite{ROS}. Moreover, the Nonlinear Model Predictive Controller employed in this work is designed using the NLopt Optimization library \cite{MLopt}  and runs on a desktop with $8$ cores, $3.40$GHz Intel Core TM $i7-6700$ CPU and $32$GB of RAM.

The workspace that the robot can operate in as well as a panoramic view of it are depicted in Fig. \ref{fig:nexus_lab} and Fig. \ref{fig:workspace}, respectively. It is captured by the set $\mathcal{W} = \{x \in \mathbb{R}^2 : \|x\|_{\infty} \le 2.5 \}$ and contains $9$ RoI which are divided as follows:
\begin{itemize}
\item the RoI $\mathcal{R}_i$, $i \in \{1, 2, 3, 4, 5\}$ depicted with blue in Fig. \ref{fig:workspace}; The RoI $\mathcal{R}_1$ and $\mathcal{R}_3$ map into the atomic propositions $\rm mission_1$ and $\rm mission_2$, respectively.
\item  the RoI $\mathcal{R}_i$, $i \in \{6, 7, 8, 9\}$ depicted with red in Fig. \ref{fig:workspace} which stand for unsafe regions that the robot needs to avoid. They map into the atomic propositions $\rm obs_1$, $\rm obs_2$ and $\rm obs_3$, respectively. Moreover, it holds that $L(\mathcal{R}_i) = \emptyset$, $\forall i \in \{2, 4, 5\}$.
\end{itemize}
The control input constraints are set to: $$\mathcal{U} = \{u_i \in \mathbb{R}^{3} : |u_i| \le 0.2, \ i \in \{1,2,3\}\},$$ where $u_1$, $u_2$ stand for the linear velocities and $u_3$ stands for the angular velocity. The desired MITL formula over the set of atomic propositions: $$\Gamma = \{\rm mission_1, mission_2, obs_1, obs_2, obs_3\},$$ that the robot needs to satisfy is set to:
\begin{align}
\hspace{-2mm} \varphi = \ & \square_{[0,\infty)} \{\neg {\rm obs_1} \wedge \neg {\rm obs_2} \wedge \neg {\rm obs_3} \wedge \neg {\rm obs_4} \} \notag \\
& \wedge \Diamond_{[30, 50]} \{\rm mission_2\} \wedge \Diamond_{[80, 110]} \{\rm mission_1\}, \label{eq:phi}
\end{align}

The prediction horizon is chosen as $N = 1.2 \sec$. The gains are set to $Q = P = R = 0.5 I_3$. 

\textbf{Video:}
A video demonstrating the experiment of this section can be found in the following link:

\hspace{12mm} { \tt \small \href{https://youtu.be/FcB8Pp5lQpw}{https://youtu.be/FcB8Pp5lQpw} }

By employing Algorithm $1$ the transition times between the RoI that the robot is following are given as follows:
\begin{align*}
\mathfrak{t}(\mathcal{R}_1, \mathcal{R}_2) &  = 17.3, \ \  \mathfrak{t}(\mathcal{R}_2, \mathcal{R}_3) = 20.1, \ \ \mathfrak{t}(\mathcal{R}_3, \mathcal{R}_4) = 18.2, \\
\mathfrak{t}(\mathcal{R}_4, \mathcal{R}_5) & = 18.5, \ \  \mathfrak{t}(\mathcal{R}_5, \mathcal{R}_1) = 15.7.
\end{align*}
The evolution of the states $x$, $y$ as well as the angle of the robot are presented in Fig. \ref{fig:state_x}, Fig. \ref{fig:state_y} and Fig. \ref{fig:angle}, respectively. The control input signals $u_1$, $u_2$ and $u_3$ are depicted in Fig. \ref{fig:control_1}, Fig. \ref{fig:control_2} and Fig. \ref{fig:control_3}, respectively. The transition times given above fulfill the constraints given by the MITL formula. Thus, it can be observed that the robot fulfills the MITL specification while all the imposed constraints from Problem \ref{problem} are satisfied.

\begin{figure}[t!]
	\centering
	\includegraphics[width = 70mm, height = 55mm]{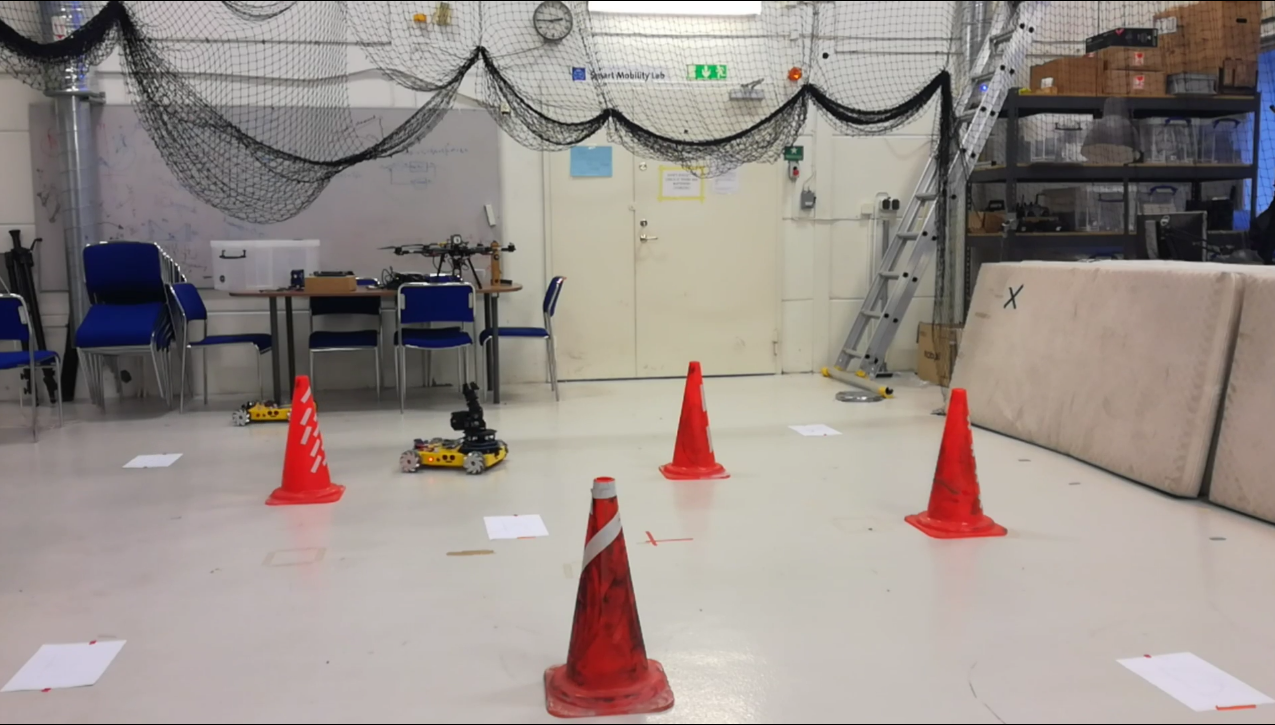}
	\caption{The nexus robot performing the desired high-level task given in \eqref{eq:phi}.}
	\label{fig:nexus_lab}
\end{figure}

\begin{figure}[t!]
	\centering
	\begin{tikzpicture}[scale = 0.85]
	\draw [color = red, fill = red!20] (-5.0, 1.7) circle (0.6cm);
	\node at (-5, 1.7) {$\mathcal{R}_6$};
	
	\draw [color = red, fill = red!20] (-2.0, 3.0) circle (0.6cm);
	\node at (-2, 3.0) {$\mathcal{R}_7$};
	
	\draw [color = red, fill = red!20] (-3.0, -1.0) circle (0.6cm);
	\node at (-3, -1.0) {$\mathcal{R}_8$};
	
	\draw [color = red, fill = red!20] (0.0, 1.0) circle (0.6cm);
	\node at (0.0, 1.0) {$\mathcal{R}_9$};
	
	\draw [color = blue, fill = blue!20] (-6.4, -1.5) circle (0.6cm);
	\node at (-6.4, -1.5) {\small $\mathcal{R}_1$};
	
	\draw [color = blue, fill = blue!20] (-6.0, 3.5) circle (0.6cm);
	\node at (-6.0, 3.5) {\small $\mathcal{R}_2$};
	
	\draw [color = blue, fill = blue!20] (-3.0, 1.5) circle (0.6cm);
	\node at (-3.0, 1.5) {$\mathcal{R}_3$};
	
	\draw [color = blue, fill = blue!20] (0.0, 3.6) circle (0.6cm);
	\node at (0.0, 3.6) {$\mathcal{R}_4$};
	
	\draw [color = blue, fill = blue!20] (0.0, -1.3) circle (0.6cm);			
	\node at (0.0, -1.3) {\small $\mathcal{R}_5$};
	
	\draw[draw = black] (-7.2, -2.5) rectangle ++(7.9, 6.9);
	\end{tikzpicture}
	\caption{A panoramic view of the workspace with the $9$ RoI.}
	\label{fig:workspace}
\end{figure}
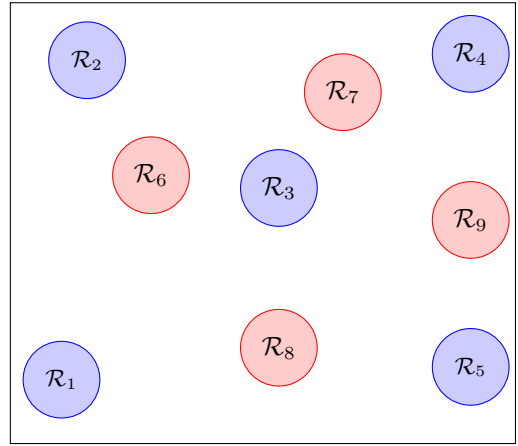

\begin{figure}[t!]
\centering
\includegraphics[scale = 0.40]{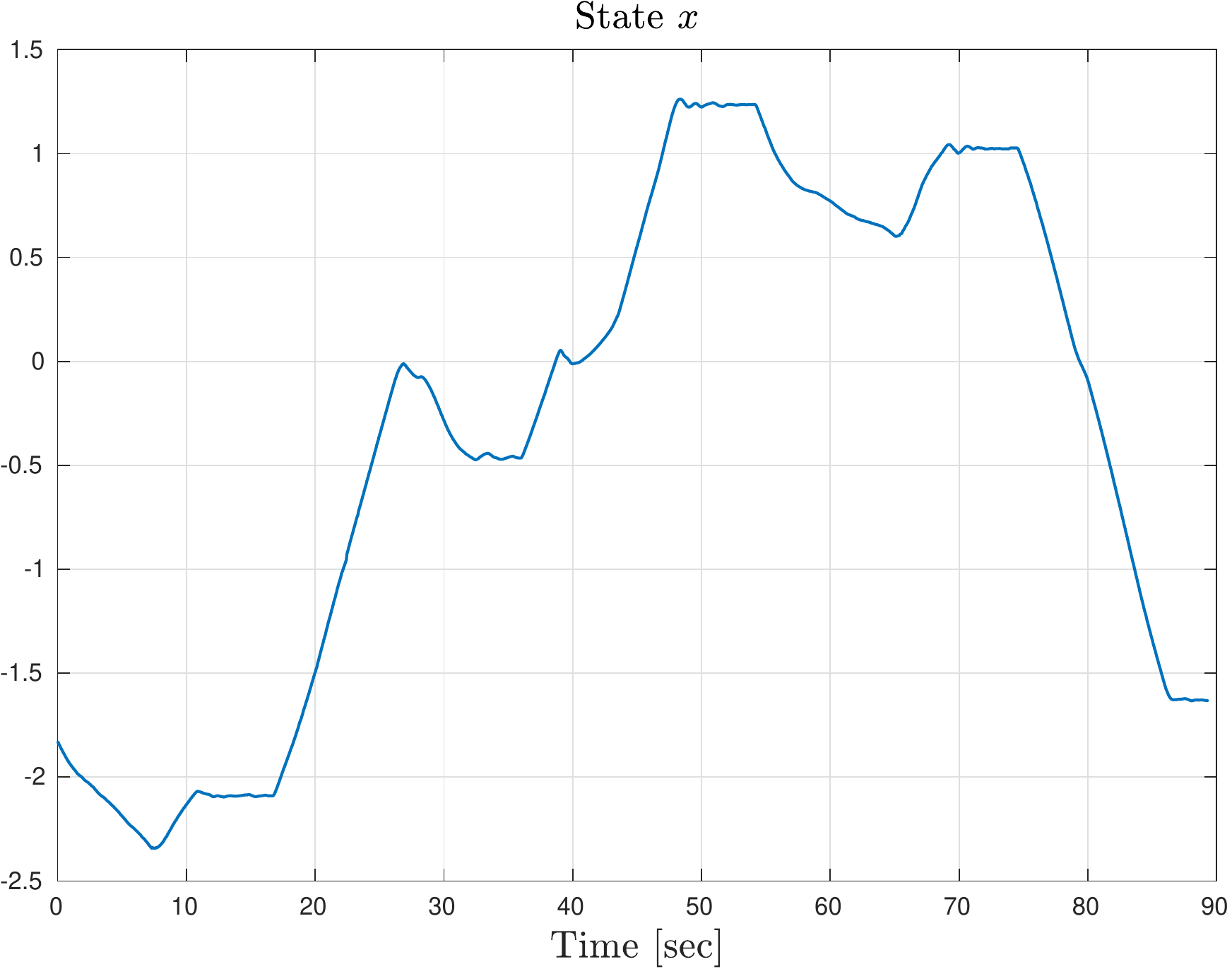}
\caption{The state $x$ of the robot.}
\label{fig:state_x}
\end{figure}

\begin{figure}[t!]
\centering
\includegraphics[scale = 0.40]{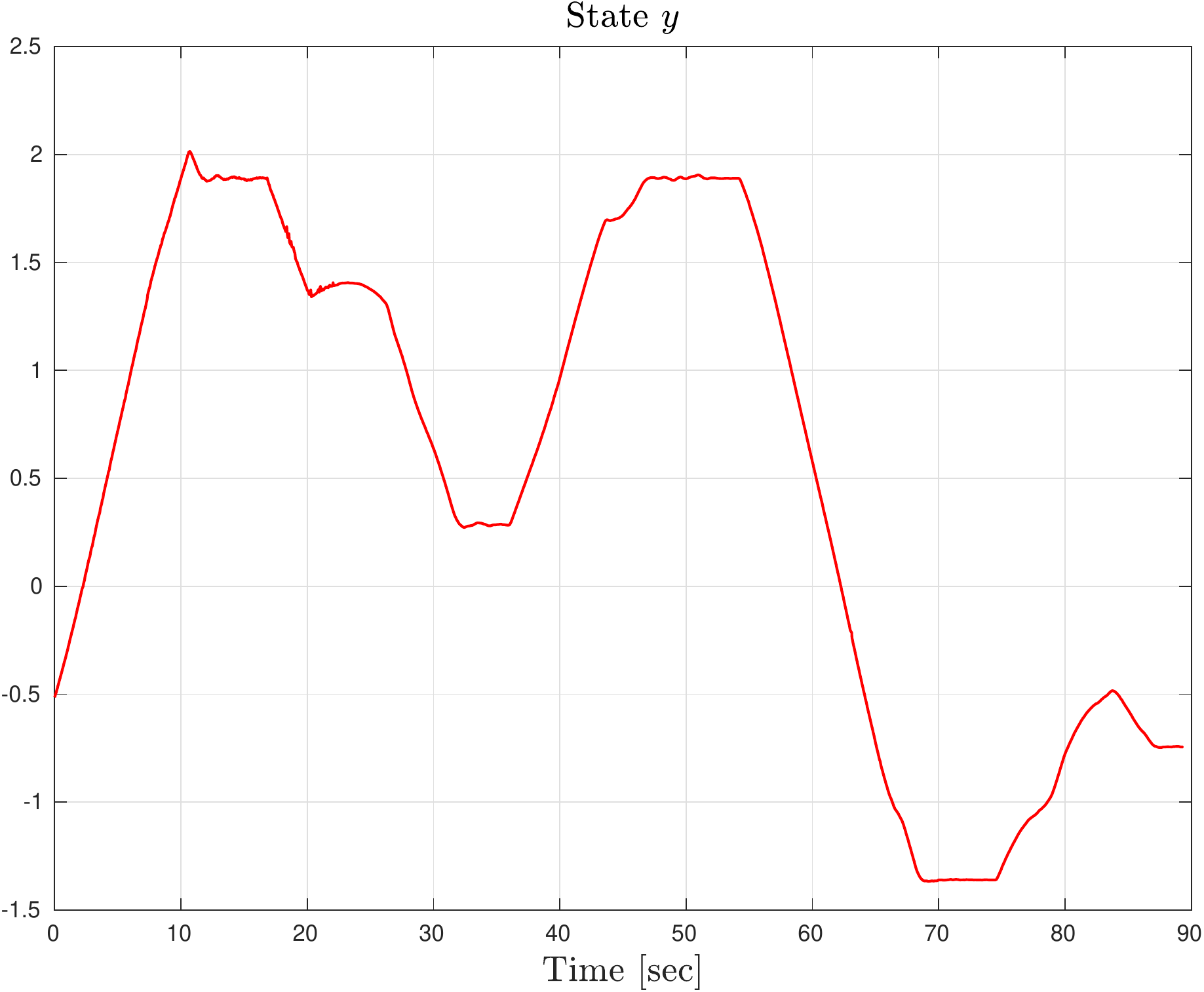}
\caption{The evolution of the state $y$ of the robot over time.}
\label{fig:state_y}
\end{figure}

\begin{figure}[t!]
\centering
\includegraphics[scale = 0.45]{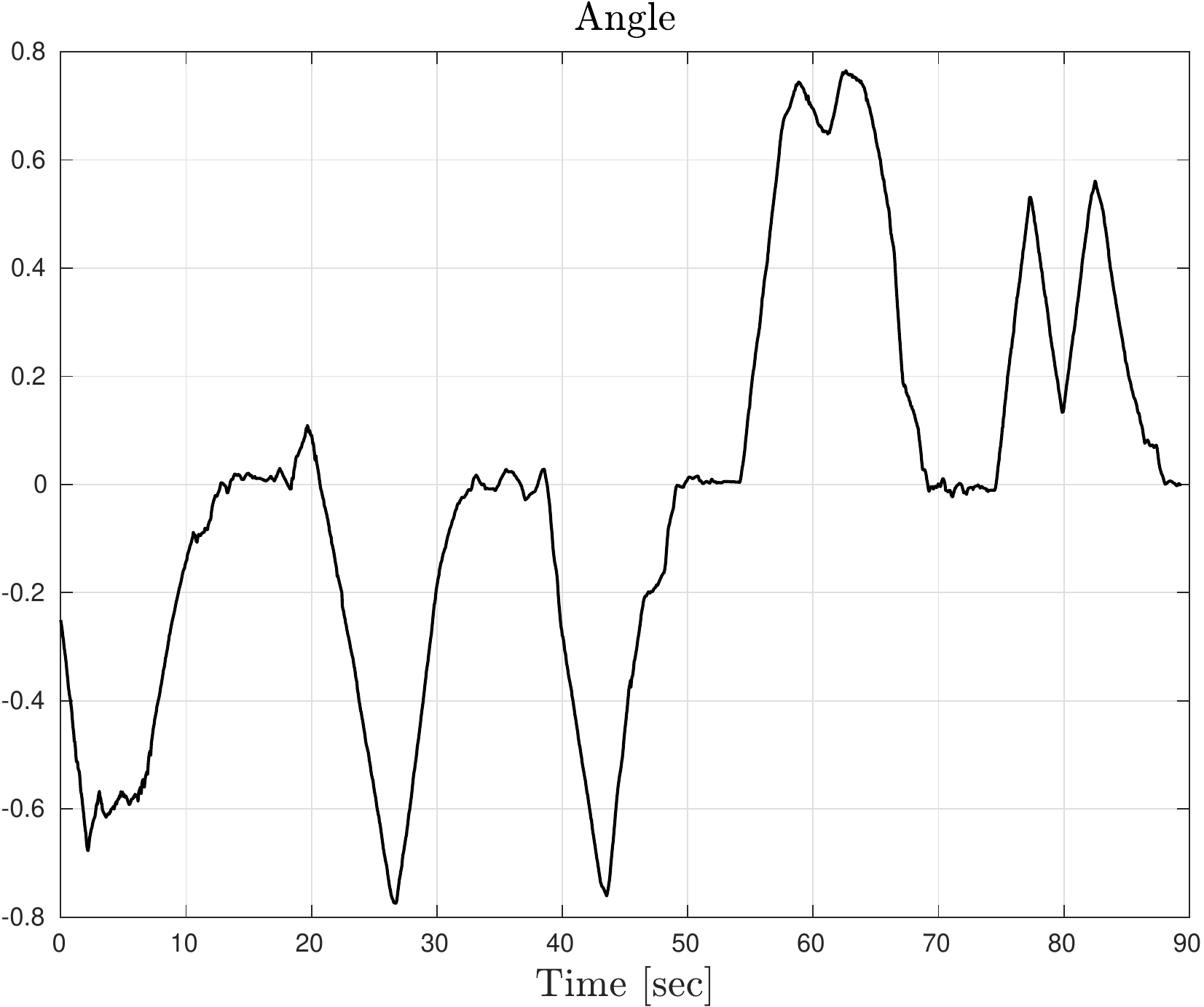}
\caption{The evolution of the angle of the robot over time.}
\label{fig:angle}
\end{figure}

\begin{figure}[t!]
\centering
\includegraphics[scale = 0.40]{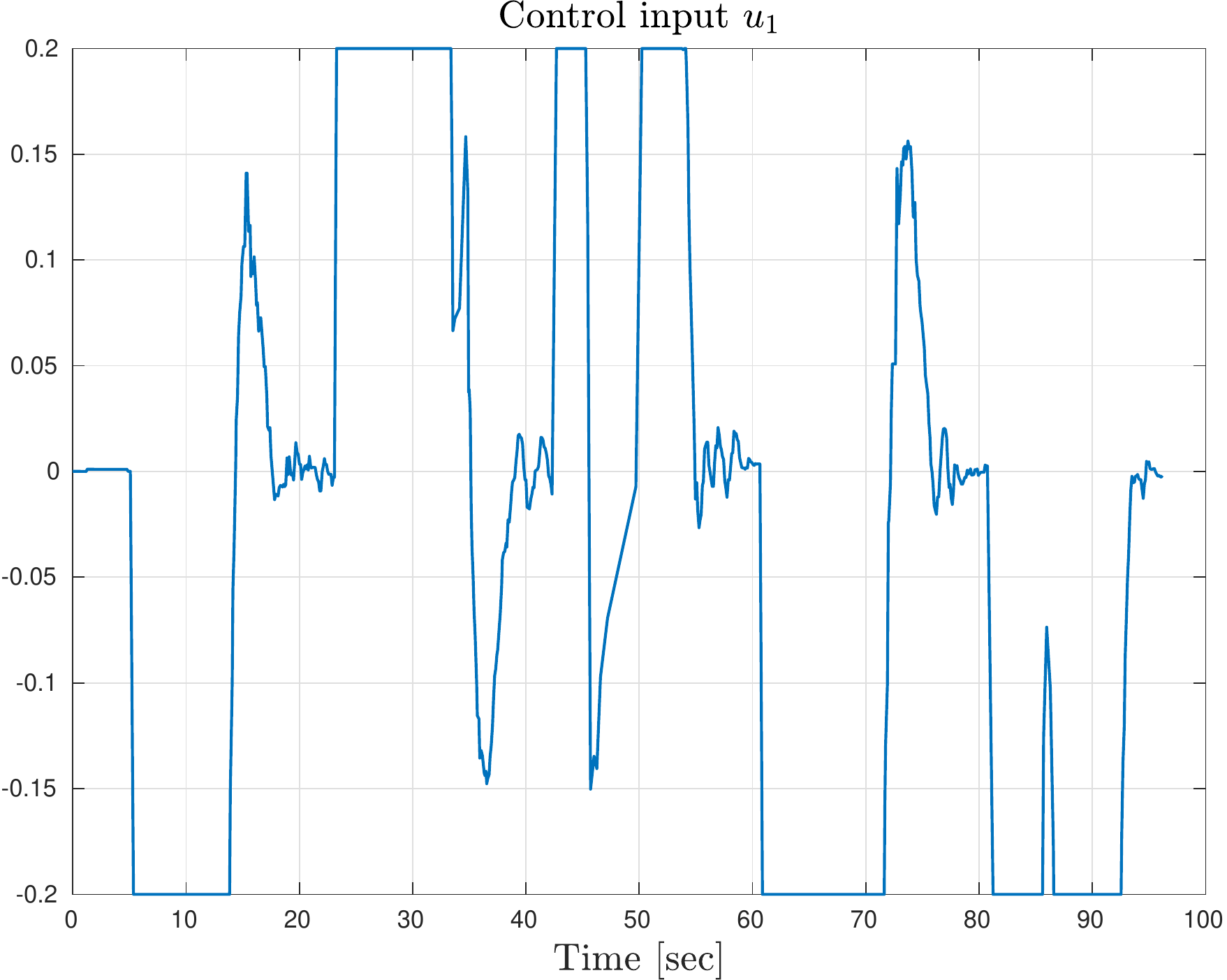}
\caption{The control input signal $u_1$ of the robot.}
\label{fig:control_1}
\end{figure}

\begin{figure}[t!]
\centering
\includegraphics[scale = 0.40]{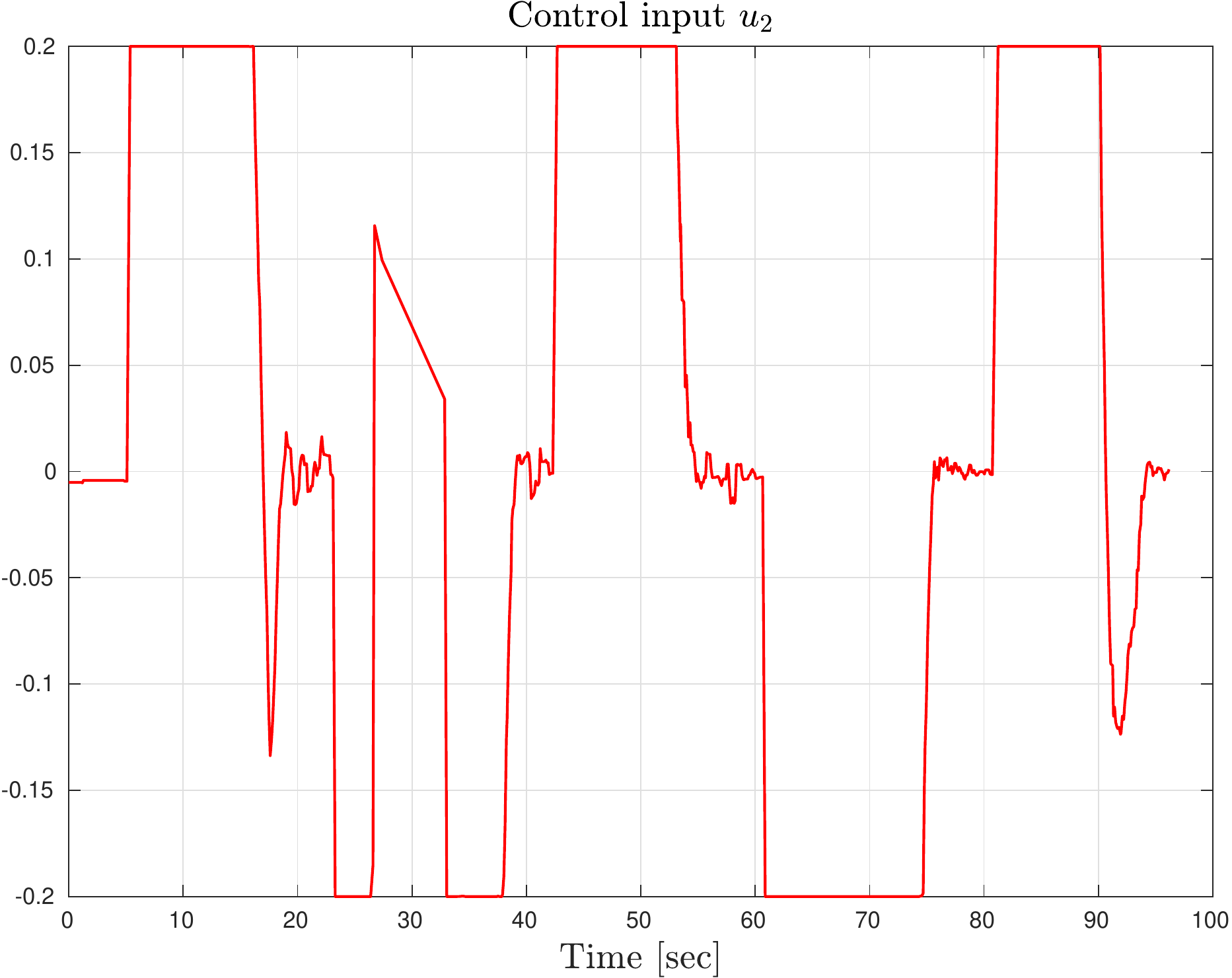}
\caption{The evolution of the control input signal $u_2$ of the robot.}
\label{fig:control_2}
\end{figure}

\begin{figure}[t!]
\centering
\includegraphics[scale = 0.45]{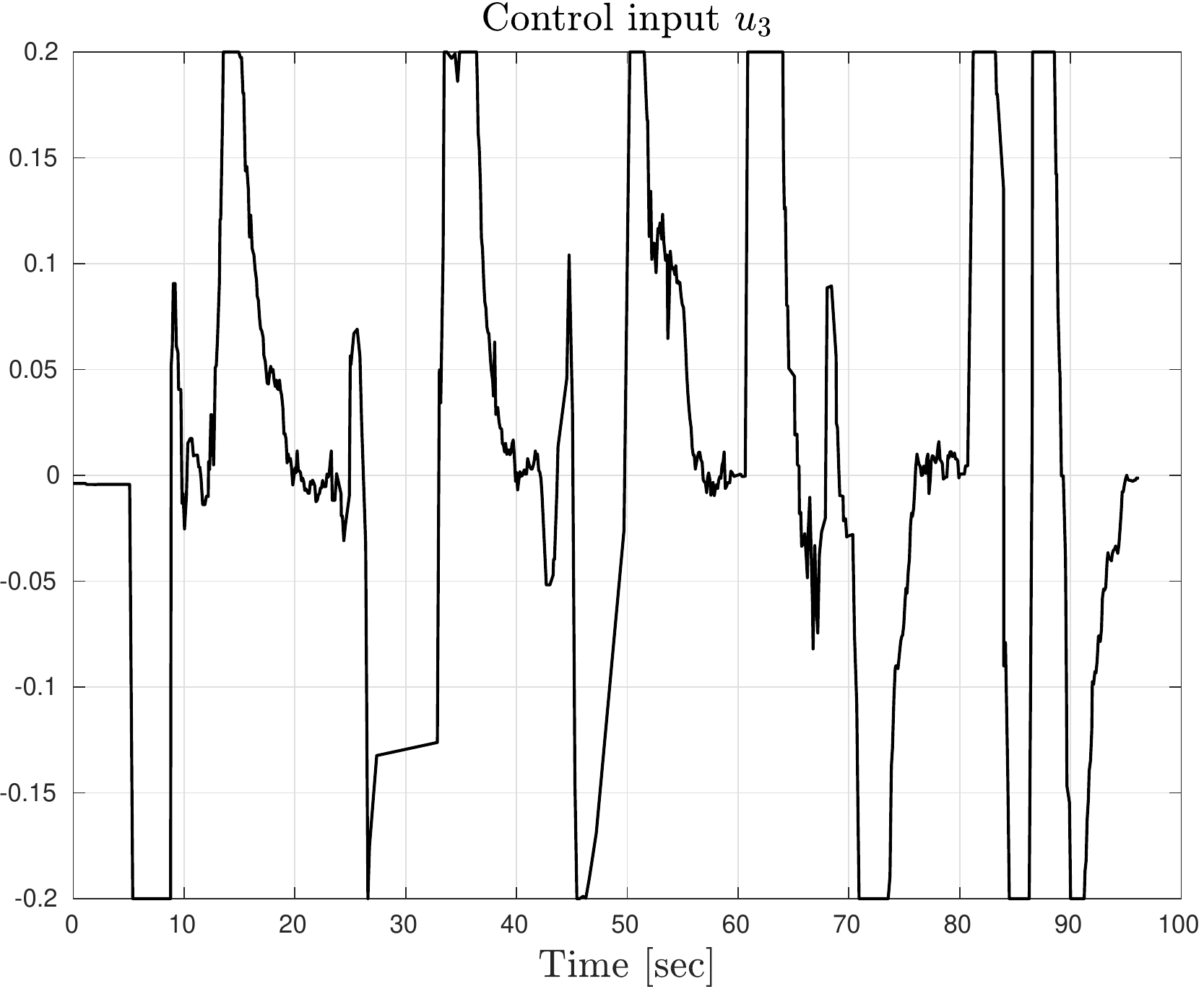}
\caption{The evolution of the control input signal $u_3$ of the robot.}
\label{fig:control_3}
\end{figure}

\section{Conclusions} \label{sec:conclusions}

In this paper, we have experimentally validated recent theoretical results of robust nonlinear tube-based MPC along with timed-constrained high-level planning. In particular, given the uncertain dynamics of a robot and a timed specification written in MITL, we have provided a framework under which a sequence of control laws under which the robot satisfies the desired task. The experimental platform consists of a Nexus $10011$ mobile robot with an attached manipulator. The preliminary experimental results of the paper in hand verifies the efficiency of the proposed framework that solves a general category of time-constrained robot navigation problems under state/input constraints, obstacle avoidance as well as uncertainties/disturbances.

\bibliographystyle{ieeetr}        
\bibliography{references}

\begin{thebibliography}{10}

\bibitem{alur_1994}
R.~Alur, T.~Feder, and T.~Henzinger, ``{T}he {B}enefits of {R}elaxing
  {P}unctuality,'' {\em Journal of the Association for Computing Machinery
  (JACM)}, vol.~43, no.~1, pp.~116--146, 1994.

\bibitem{alex_2016_acc}
A.~Nikou, J.~Tumova, and D.~V. Dimarogonas, ``{C}ooperative {T}ask {P}lanning
  of {M}ulti-{A}gent {S}ystems {U}nder {T}imed {T}emporal {S}pecifications,''
  {\em American Control Conference (ACC)}, pp.~13--19, 2016.

\bibitem{alex_licenciate}
A.~Nikou, ``{C}ooperative {P}lanning {C}ontrol and {F}ormation {C}ontrol of
  {M}ulti-{A}gent {S}ystems {M}ulti-{A}gent {S}ystems,'' {\em Licentiate
  Thesis, KTH Royal Institute of Technology, Link:
  \href{http://www.diva-portal.org/smash/get/diva2:1094985/FULLTEXT01.pdf}{http://www.diva-portal.org/smash/get/diva2:1094985/FULLTEXT01.pdf}},
  2017.

\bibitem{alex_automatica_2018}
A.~Nikou, D.~Boskos, J.~Tumova, and D.~V. Dimarogonas, ``{O}n the {T}imed
  {T}emporal {L}ogic {P}lanning of {C}oupled {M}ulti-{A}gent {S}ystems,'' {\em
  Automatica}, vol.~97, pp.~339--345, 2018.

\bibitem{KODITSCHEK1990412}
D.~Koditschek and E.~Rimon, ``{R}obot {N}avigation {F}unctions on {M}anifolds
  with {B}oundary,'' {\em Advances in Applied Mathematics}, vol.~11, no.~4,
  pp.~412 -- 442, 1990.

\bibitem{babis_automatica}
C.~Bechlioulis and G.~Rovithakis, ``{A} {L}ow-{C}omplexity {G}lobal
  {A}pproximation-{F}ree {C}ontrol {S}cheme with {P}rescribed {P}erformance for
  {U}nknown {P}ure {F}eedback {S}ystems,'' {\em Automatica}, vol.~50, no.~4,
  pp.~1217 -- 1226, 2014.

\bibitem{rakovic_2004_tubes_1}
W.~Langson, I.~Chryssochoos, S.~Rakovic, and D.~Mayne, ``{R}obust {M}odel
  {P}redictive {C}ontrol {U}sing {T}ubes,'' {\em Automatica}, vol.~40, no.~1,
  pp.~125--133, 2004.

\bibitem{yu_2013_tube}
S.~Yu, C.~Maier, H.~Chen, and F.~Allg{\"o}wer, ``{T}ube {MPC} {S}cheme {B}ased
  on {R}obust {C}ontrol {I}nvariant {S}et with {A}pplication to {L}ipschitz
  {N}onlinear {S}ystems,'' {\em Systems and Control Letters}, vol.~62, no.~2,
  pp.~194--200, 2013.

\bibitem{alex_ACC_2018}
A.~Nikou and D.~V. Dimarogonas, ``{R}obust {T}ube-based {M}odel {P}redictive
  {C}ontrol for {T}imed-constrained {R}obot {N}avigation,'' {\em American
  Control Conference (ACC),
  \href{https://arxiv.org/pdf/1809.09825.pdf}{https://arxiv.org/pdf/1809.09825.pdf}},
  Philadelphia, US, 2019.

\bibitem{alex_IJRNC_2018}
A.~Nikou and D.~V. Dimarogonas, ``{D}ecentralized {T}ube-based {M}odel
  {P}redictive {C}ontrol of {U}ncertain {N}onlinear {M}ulti-{A}gent
  {S}ystems,'' {\em International Journal of Robust and Nonlinear Control
  (IJRNC)}, pp.~1--20,
  \href{https://doi.org/10.1002/rnc.4522}{https://doi.org/10.1002/rnc.4522},
  2019.

\bibitem{shaha1}
S.~Heshmati-alamdari, G.~Karras, P.~Marantos, and K.~Kyriakopoulos, ``{A}
  {R}obust {M}odel {P}redictive {C}ontrol {A}pproach for {A}utonomous
  {U}nderwater {V}ehicles {O}perating in a {C}onstrained {W}orkspace,'' {\em
  2018 IEEE International Conference on Robotics and Automation (ICRA)},
  pp.~1--5, 2018.

\bibitem{shaha2}
S.~Heshmati-alamdari, G.~Karras, and K.~Kyriakopoulos, ``{A} {D}istributed
  {P}redictive {C}ontrol {A}pproach for {C}ooperative {M}anipulation of
  {M}ultiple {U}nderwater {V}ehicle {M}anipulator {S}ystems,'' {\em 2019 IEEE
  International Conference on Robotics and Automation (ICRA)}, 2019.

\bibitem{shaha3}
S.~Heshmati-alamdari, G.~Karras, A.~Eqtami, and K.~Kyriakopoulos, ``{A}
  {R}obust {S}elf {T}riggered {I}mage {B}ased {V}isual {S}ervoing {M}odel
  {P}redictive {C}ontrol {S}cheme for {S}mall {A}utonomous {R}obots,'' {\em
  2015 IEEE/RSJ International Conference on Intelligent Robots and Systems
  (IROS)}, pp.~5492--5497, 2015.

\bibitem{alex_shahab_ifac}
S.~Heshmati-alamdari, A.~Nikou, K.~J. Kyriakopoulos, and D.~V. Dimarogonas,
  ``{A} {R}obust {F}orce {C}ontrol {A}pproach for {U}nderwater {V}ehicle
  {M}anipulator {S}ystems,'' {\em 20th World Congress of the International
  Federation of Automatic Control (IFAC WC), Toulouse, France}, vol.~50, Issue
  1, pp.~11197--11202, 2017.

\bibitem{bouyer2009qualitative}
P.~Bouyer, ``{F}rom {Q}ualitative to {Q}uantitative {A}nalysis of {T}imed
  {S}ystems,'' {\em M{\'e}moire d’habilitation, Universit{\'e} Paris},
  vol.~7, 2009.

\bibitem{frank_1998_quasi_infinite}
H.~Chen and F.~Allg{\"o}wer, ``{A} {Q}uasi-{I}nfinite {H}orizon {N}onlinear
  {M}odel {P}redictive {C}ontrol {S}cheme with {G}uaranteed {S}tability,'' {\em
  Automatica}, vol.~34, no.~10, pp.~1205--1217, 1998.

\bibitem{SML}
``Smart mobility lab ({SML}),'' {\em
  \href{https://www.kth.se/dcs/research/control-of-transport/smart-mobility-lab/smart-mobility-lab-1.441539}{https://www.kth.se/dcs/research/control-of-transport/smart-mobility-lab/smart-mobility-lab-1.441539}}.

\bibitem{ROS}
M.~Quigley, K.~Conley, B.~P. Gerkey, J.~Faust, T.~Foote, J.~Leibs, R.~Wheeler,
  and A.~Y. Ng, ``{ROS}: {A}n {O}pen-{S}ource {R}obot {O}perating {S}ystem,''
  {\em International Conference on Robotics and Automation (ICRA) Workshop},
  2009.

\bibitem{MLopt}
S.~G. Johnson, ``The {NL}opt {N}onlinear-{O}ptimization {P}ackage,'' {\em
  \href{http://ab-initio.mit.edu/nlopt}{http://ab-initio.mit.edu/nlopt}}, 2009.

\end{thebibliography}
\end{document}